\documentclass[11pt,letterpaper]{article}

\usepackage{fullpage}
\usepackage{hyperref}
\usepackage[numbers]{natbib}
\usepackage{xspace}
\usepackage{amsmath,amssymb,amsthm,amsfonts}
\usepackage{mathrsfs}
\usepackage{dsfont}
\usepackage[figure,linesnumbered]{algorithm2e}

\usepackage{tikz}

\newtheorem{theorem}{Theorem}
\newtheorem{lemma}{Lemma}
\newtheorem{corollary}{Corollary}

\newcommand{\email}[1]{\href{mailto:#1}{\nolinkurl{#1}}}

\newcommand{\secref}[1]{Section~\ref{#1}}
\newcommand{\thmref}[1]{Theorem~\ref{#1}}
\newcommand{\lemref}[1]{Lemma~\ref{#1}}
\newcommand{\figref}[1]{Figure~\ref{#1}}
\newcommand{\lineref}[1]{Line~\ref{#1}}

\newcommand{\ie}{i.e.,\xspace}
\newcommand{\eg}{e.g.,\xspace}

\renewcommand{\deg}[2]{\delta^{-}_{#1}(#2)}
\newcommand{\degb}[2]{\delta^{-}_{#1}\bigl(#2\bigr)}
\newcommand{\degp}[1]{\delta^{+}(#1)}
\newcommand{\Deg}{\Delta}

\renewcommand{\vec}[1]{{\boldsymbol{#1}}}

\newcommand{\A}{\vec{A}}
\newcommand{\E}{\mathbb{E}}
\newcommand{\G}{\mathcal{G}}
\newcommand{\N}{\mathbb{N}}
\newcommand{\Q}{\mathbb{Q}}
\renewcommand{\P}{\mathbb{P}}
\newcommand{\Part}{\mathcal{P}}
\renewcommand{\part}{P}
\renewcommand{\S}{\vec{S}}
\newcommand{\Scal}{\mathcal{S}}
\newcommand{\midd}{:}
\newcommand{\cond}{\,|\,}
\newcommand{\imax}{i^*}
\newcommand{\dmax}{d^*}
\newcommand{\xmax}{\hat{x}}
\newcommand{\xmin}{\check{x}}

\newcommand{\chr}[1]{\ifthenelse{\equal{#1}{}}{\chi}{\chi(#1)}}
\newcommand{\chrb}[1]{\chi\bigl(#1\bigr)}

\newlength{\wordlength}

\def\clap#1{\hbox to 0pt{\hss#1\hss}}
 \def\mathrlap{\mathpalette\mathrlapinternal} \def\mathclap{\mathpalette\mathclapinternal}

\def\mathrlapinternal#1#2{\rlap{$\mathsurround=0pt#1{#2}$}}
\def\mathclapinternal#1#2{\clap{$\mathsurround=0pt#1{#2}$}}

\begin{document}

\title{Optimal Impartial Selection\thanks{Part of the work was done while the first author enjoyed the hospitality of the Combinatorial Optimization and Graph Algorithms group at TU Berlin. Valuable discussions with Paul D\"utting and Frank Kelly are gratefully acknowledged.}}

\author{%
	Felix Fischer\thanks{%
	Statistical Laboratory,
	University of Cambridge,
	Wilberforce Road, Cambridge CB3 0WB, UK.
	Email: \email{fischerf@statslab.cam.ac.uk}.}
	\and Max Klimm\thanks{%
	Technische Universit\"at Berlin,
	Institut f\"ur Mathematik, 
	Stra\ss{}e des 17.~Juni 136, 10623 Berlin, Germany. 
	Email: \email{klimm@math.tu-berlin.de}}}
	
\date{}
	
\maketitle

\begin{abstract}
	We study the problem of selecting a member of a set of agents based on impartial nominations by agents from that set. The problem was studied previously by \citeauthor{AFPT11a} and \citeauthor{HoMo13a} and has important applications in situations where representatives are selected from within a group or where publishing or funding decisions are made based on a process of peer review. Our main result concerns a randomized mechanism that in expectation awards the prize to an agent with at least half the maximum number of nominations. Subject to impartiality, this is best possible.
\end{abstract}

\section{Introduction}

We consider a situation where members of a set of agents nominate other agents from the set for a prize and the goal is to award the prize to an agent who receives a large number of nominations. This situation arises naturally, for example, when representatives are selected from within a group or when publishing or funding decisions are made based on a process of peer review. While nominations are at the discretion of the nominating agents, it is often reasonable to assume that agents are impartial to the selection of others and will nominate who they think should receive the prize as long as they cannot influence their own chances of receiving it. Indeed, the assumption of impartiality was previously made, and justified, in the very same setting~\citep{AFPT11a,HoMo13a}.

Formally, the situation can be captured by a directed graph with~$n$ vertices, one for each agent, in which edges correspond to nominations. A selection mechanism then chooses a vertex for any given graph, and impartiality requires that the chances of a particular vertex to be chosen do not depend on its outgoing edges.
It is easy to see that an impartial mechanism cannot always select a vertex with maximum indegree, corresponding to an agent with a maximum number of nominations, even when $n=2$. We therefore aim at maximizing the indegree of the selected vertex relative to the maximum indegree, and call a mechanism $\alpha$-optimal, for $\alpha\leq 1$, if for any graph the former is at least $\alpha$ times the latter. We focus here on the selection of a single agent, but note that it is an interesting question whether optimal mechanism for the general case can be obtained directly from mechanisms for selecting a single agent or whether their design requires additional techniques.

\paragraph{State of the Art}

\citet{AFPT11a} and \citet{HoMo13a} showed independently that deterministic impartial mechanisms are extremely limited, and
must sometimes select an agent with zero nominations even though agents are being nominated, or an agent with one nomination
when another agent receives $n-1$ nominations.

On the other hand, \citeauthor{AFPT11a} proposed a simple randomized mechanism that partitions the agents into two sets $S_1$ and $S_2$ and selects an agent from $S_2$ who among agents in this set receives a maximum number of nominations by agents in $S_1$. By linearity of expectation the mechanism is at least $1/4$-optimal, and a situation with a single nomination shows that it cannot do better. A somewhat closer inspection of situations with one or two nominations shows that no impartial mechanism can be better than $1/2$-optimal. While these bounds are almost trivial, no improvements have been obtained that hold for general values of $n$, despite considerable efforts. Moreover, improving the lower bound appears just as difficult in the special case where each agent submits exactly one nomination, as considered by \citeauthor{HoMo13a}. This is somewhat embarrassing, as the mechanism of \citeauthor{AFPT11a} should intuitively be better than $1/4$-optimal as soon as there is more than just a single nomination.

\paragraph{Our Contribution}

\citeauthor{AFPT11a}'s analysis of the $2$-partition mechanism is tight and yields a constant approximation ratio, only a factor of two away from the best possible one. Quite strikingly, however, the analysis requires almost no understanding of the mechanism or the problem the mechanism is trying to solve. As a consequence it does not lead to stronger bounds for special cases, like the setting with one nomination per agent studied by \citeauthor{HoMo13a}, and cannot be extended to more complicated mechanisms.

Our first result attempts to close this gap in our understanding of the $2$-partition mechanism, by providing a lower bound on its performance relative to the maximum indegree. Among other things, this yields a lower bound of $3/8$ for the case where each agent submits at least one nomination. Our analysis uses a novel adversarial argument that allows us to abstract from the underlying graph structure and isolate the critical aspects of difficult problem instances.

More interestingly, our analysis suggests a natural generalization of the $2$-partition mechanism that partitions the set of agents into $k>2$ sets and iteratively considers the nominations submitted by agents in more and more of these sets, to fewer and fewer candidates in the remaining sets. Intuitively this increases the probability of each individual nomination to be counted, which is particularly important in the difficult cases with a small overall number of nominations. Exactly how information from an earlier stage of the mechanism can be used without a negative effect on later stages turns out to be somewhat intricate.

We then generalize the adversarial analysis to show that the $k$-partition mechanism provides an approximation ratio of $(k-1)/2k$, which approaches the upper bound of $1/2$ as $k$ tends to infinity. This implicitly provides an analysis of a limiting mechanism, in which agents are considered one by one according to a random permutation.

We finally give the first non-trivial bounds for settings without abstentions, where the permutation mechanism is at least $7/12$-optimal and at most $2/3$-optimal, and no impartial mechanism can be more than $3/4$-optimal. Quite intriguingly, the exact upper bounds approach $3/4$ from below as the number of agents grows and are tight for small numbers of agents. This can be seen as evidence that optimal mechanisms in this case might be rather difficult to find.

\paragraph{Related Work and Applications}

Impartial decision making was first considered by \citet{CMT08a}, for the case of a divisible resource to be shared among a set of agents. While the difference between a divisible resource and the indivisible resource considered in this paper disappears for randomized mechanisms, \citeauthor{CMT08a} studied mechanisms with a more general message space that allows for fractional nominations, and at the same time aimed for weaker requirements to be achieved besides impartiality.

\citet{AFPT11a} framed the problem considered here as one of designing approximately optimal strategyproof mechanisms without payments, an agenda proposed by~\citet{PrTe13a} and earlier by~\citet{DFP10a}. Strategyproofness requires that an agent maximizes its utility by truthfully revealing its preferences and is equivalent to impartiality if the utility of an agent only depends on its chances of being selected. While this assumption seems somewhat restrictive, \citeauthor{AFPT11a} pointed out that their results in fact hold for any setting where agents give their own selection priority over that of their nominees.

Strategyproof selection is an important component of the peer review process for scientific articles and project proposals. For its Sensors and Sensing Systems program, the National Science Foundation (NSF) recently introduced a mechanism in which proposals are reviewed by other applicants, and acceptance of an applicant's own proposal depends in part on the extent to which the \emph{reviews} submitted by the applicant agree with other reviews of the same proposals. The specific mechanism used by the NSF was originally proposed by \citet{MeSa09a} in the context of allocation of telescope time. Whether the mechanism provides the right incentives in peer review is debatable, but its lack of impartiality, which in this case is deliberate, would make it very hard to show any formal incentive properties. By contrast, our results allow for a separation of preferences regarding an agent's own selection and those regarding the selection of others, and can in principle be combined with peer prediction techniques~\citep[\eg][]{WiPa12a} to provide strict incentives for the truthful evaluation of other agents. The exact properties achievable by such hybrid mechanisms deserve further investigation.

Impartial selection is also more distantly related to work in distributed computing on leader election~\citep[\eg][]{AlNa93a,CoLi95a,Feig99a,Anto06a} and work on the manipulation of reputation systems~\citep[\eg][]{FRS07a}. Leader election seeks to guarantee the selection of a non-malicious agent in the presence of malicious agents trying to manipulate the selection process. Work on reputation systems often considers models with more complex preference and message spaces, where maximization of a one-dimensional objective does not suffice.

The $2$-partition mechanism, finally, is reminiscent of random sampling in unlimited-supply auctions~\citep{FGHK02a,GHKo06a,FFHK05a} and combinatorial auctions~\citep{DNS06a}. It will be interesting to see whether our more complicated mechanisms and analysis techniques can be applied to these settings in a meaningful way.

\section{Preliminaries}

For $n\in\N$, let $\G_n=\{(N,E)\midd N=\{1,\dots,n\}, E\subseteq (N\times N)\setminus\bigcup_{i\in N}(\{i\}\times\{i\})\}$ be the set of directed graphs with $n$ vertices and no loops. Let $\G=\bigcup_{n\in\N}\G_n$. For $G=(N,E)\in\G$, $S\subseteq N$, and $i\in N$, let $\deg{S}{i,G}=|\{(j,i)\in E\midd G=(N,E),j\in S\}|$ denote the indegree of vertex $i$ from vertices in $S$. We use $\deg{}{i,G}$ as a shorthand for $\deg{N}{i,G}$, denote $\Deg(G)=\max_{i\in N}\deg{}{i,G}$, and write $\deg{}{i}$ instead of $\deg{}{i,G}$ and $\Deg$ instead of $\Deg(G)$ if $G$ is clear from the context.

A \emph{selection mechanism} for $\G$ is then given by a family of functions $f:\G_n\rightarrow[0,1]^n$ that maps each graph to a probability distribution on its vertices. In a slight abuse of notation, we use $f$ to refer to both the mechanism and individual functions from the family. 
Mechanism~$f$ is \emph{impartial} on $\G'\subseteq\G$ if on this set of graphs the probability of selecting vertex~$i$ does not depend on its outgoing edges, \ie if for every pair of graphs $G=(N,E)$ and $G'=(N,E')$ in $\G'$ and every $i\in N$, $(f(G))_i=(f(G'))_i$ whenever $E\setminus(\{i\}\times V)=E'\setminus(\{i\}\times V)$. All mechanisms we consider are impartial on $\G$, and we simply refer to such mechanisms as impartial mechanisms. 
Mechanism~$f$ is \emph{$\alpha$-optimal} on $\G'\subseteq\G$, for $\alpha\leq 1$, if for any graph in $\G'$ the expected degree of the vertex selected by $f$ differs from the maximum degree by a factor of at most $\alpha$, \ie if
\[
	\inf_{\mathclap{\substack{G\in\G\\\Delta(G)>0}}} \frac{\E_{i\sim f(G)}[\deg{}{i,G}]}{\Deg(G)}\geq \alpha .
\]
We call a mechanism $\alpha$-optimal if it is~$\alpha$-optimal on~$\G$, and approximately optimal if it $\alpha$-optimal for some constant~$\alpha$.

As far as impartiality and approximate optimality are concerned, we can restrict our attention to symmetric mechanisms. Mechanism $f$ is \emph{symmetric} if it is invariant with respect to renaming of the vertices, \ie if for every $G=(N,E)\in\G$, every $i\in N$, and every permutation $\pi:N\rightarrow N$,
\[
	(f(G_\pi))_{\pi(i)} = (f(G))_i ,
\]
where $G_\pi=(N,E_\pi)$ with $E_\pi=\{(\pi(i),\pi(j))\midd (i,j)\in E\}$.
For a given mechanism $f$, denote by $f_s$ the mechanism obtained by applying a random permutation $\pi$ to the vertices to the input graph, invoking $f$, and permuting the result by the inverse of $\pi$, such that for all $n\in\N$, $G\in\G_n$, and $i\in\{1,\dots,n\}$,
\[
	\bigl(f_s(G)\bigr)_i = \frac{1}{n!} \sum_{\mathclap{\pi\in\Scal_n}} \bigl(f(G_\pi)\bigr)_{\pi_i} ,
\]
where $\Scal_n$ is the set of all permutations $\pi=(\pi_1,\dots,p_n)$ of a set of $n$ elements.
The following result is straightforward.
\begin{lemma}[\citet{HoMo13a}]  \label{lem:symmetric}
	Let $f$ be a selection mechanism that is impartial and $\alpha$-optimal on $\G'\subseteq\G$. Then $f_s$ is impartial, $\alpha$-optimal, and symmetric on $\G'$.
\end{lemma}

\section{The $2$-Partition Mechanism}
\label{sec:2_partition}

We begin our investigation with a more detailed analysis of the $2$-partition mechanism proposed by \citet{AFPT11a}. The mechanism first assigns each vertex
independently and uniformly at random to one of two sets $A_1$ and $A_2$, such that $A_1\cup A_2=N$, $A_1\cap A_2=\emptyset$, and $\P[i\in A_1] = \P[i\in A_2]=1/2$ for all $i\in N$. Then it returns a vertex from $A_2$ that has maximum indegree from vertices in $A_1$, or a vertex chosen uniformly at random from $N$ in case $A_2=\emptyset$. A formal description of the mechanism is given in \figref{alg:2_partition}.
\begin{algorithm}[tb]
\KwIn{Graph $G=(N,E)$}
\KwOut{Vertex $i\in N$}
\label{line:2_partition}
Assign each $i\in N$ independently and uniformly at random to one of two sets $A_1$ and $A_2$\;
\lIf{$A_2=\emptyset$}{return a vertex chosen uniformly at random from $N$}\;
Return a vertex chosen uniformly at random from $\arg\max_{i\in A_2}\deg{A_1}{i}$\;
\caption{The $2$-partition mechanism}
\label{alg:2_partition}
\end{algorithm}

The $2$-partition mechanism is obviously impartial, as any given vertex is either in $A_1$, in which case it will never be selected, or in $A_2$, in which case its outgoing edges have no influence on the outcome of the mechanism. It is also easy to see that the mechanism is $1/4$-optimal. For an arbitrary graph $G$ and a particular vertex $\imax$ in $G$ with degree $\Delta=\Delta(G)$, we have that $\P[\imax\in A_2]=1/2$ and, by linearity of expectation, $\E[\deg{A_1}{\imax}\cond\imax\in A_2]=\E[\deg{A_1}{\imax}]=\deg{N}{\imax}/2=\Delta/2$. The expected degree of the selected vertex is thus at least $\Delta/2$ with probability at least $1/2$, \ie at least $\Delta/4$. A graph with a single edge shows that this result is in fact tight.
An upper bound on $\alpha$ for \emph{any} impartial mechanism can be obtained by considering the two graphs in \figref{fig:upper}, and the probabilities $p_1$, $p_2$, and $p_3$ with which certain vertices in these graphs are selected. Due to symmetry, which we can assume by \lemref{lem:symmetric}, $p_1=p_2$ and thus $p_1\leq 1/2$. On the other hand, $p_1=p_3$ by impartiality, so the expected degree of the vertex selected in the right graph is at most $1/2$ and the claim follows.
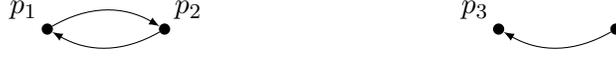
\begin{figure}[tb]
	\centering
	\begin{tikzpicture}
		\tikzstyle{every circle node}=[fill,inner sep=1.5pt,outer sep=0pt]
		\draw (0,0) +(150:.9) node(1)[circle]{} node[above left]{$p_1$} +(30:.9)
node(2)[circle]{} node[above right]{$p_2$};
		\draw (6,0) +(150:.9) node(3)[circle]{} node[above left]{$p_3$} +(30:.9)
node(4)[circle]{};
		\foreach \x/\y in {1/2,2/1,4/3} {\draw[-latex] (\x) to[bend left] (\y);}
	\end{tikzpicture}
	\caption{No impartial mechanism is more than $1/2$-optimal}
	\label{fig:upper}
\end{figure}

This rather straightforward analysis does not lead to a tight result, but it is unsatisfactory in particular because it provides no information about the performance of the mechanism on more complicated graphs, and no cues what a better mechanism might look like. We will gain both from the proof of the following lemma, which establishes a lower bound on the expected degree of the selected vertex relative to the maximum degree $\Delta(G)$.
\begin{lemma}  \label{lem:alpha_2}
	On any graph $G$ with maximum degree $\Delta=\Delta(G)$, the $2$-partition mechanism is $\alpha_2(\Delta)$-optimal, where
\vspace*{-1ex}
\begin{equation*}
	\alpha_2(\Delta) = \frac{1}{\Delta 2^{\Delta}} \sum_{k=0}^{\Delta} \binom{\Delta}{k} \cdot \min
	\biggl\{\frac{\Delta}{2},k\biggr\}.
\end{equation*}
\end{lemma}
\begin{proof}
Let $\imax\in N$ such that $\deg{}{\imax}=\Delta$, and denote by $X$ the degree of the agent selected by the $2$-partition
mechanism. Then $X$ is a random variable subject to the internal randomness of the mechanism, and we will be interested
in its expected value $\E[X]$.

Let $\vec{A}=(A_1,A_2)$ be the partition selected in \lineref{line:2_partition} of the $2$-partition mechanism in \figref{alg:2_partition}, and consider an
arbitrary set $S\subseteq N\setminus\{\imax\}$ of vertices other than $\imax$. We begin by bounding $\E[X\cond A_1\setminus\{\imax\}=S]$, \ie the expected value of $X$ given that $A_1=S$ or $A_1=S\cup\{\imax\}$.
To this end, let $z(S)$ and $a(S)$ denote the indegree of $\imax$ from $S$ and the maximum degree of any \emph{other} element of $N\setminus S$ from $S$, \ie $z(S)=\deg{S}{\imax}$ and $a(S)=\max\nolimits_{i\in N\setminus(S\cup\{\imax\})}\deg{S}{i}$.

Assume for now that $S\neq\emptyset$ and $S\neq N\setminus\{\imax\}$. Then, $\E[X\cond A_1=S]=\Delta$ if $z(S)>a(S)$, $\E[X\cond A_1=S]\geq a(S)$ if $a(S)\geq z(S)$, and $\E[X\cond A_1=S\cup\{\imax\}]\geq a(S)$. Note here that the expected value of $X$ only increases if there is an edge from $\imax$ to a vertex for which $a(S)$ is attained.
Since the events where $A_1=S$ and $A_1=S\cup\{\imax\}$ occur with equal probability,
\begin{align*}
	\E[X\cond A_1\setminus\{\imax\}=S] &\geq \frac{\chrb{z(S)>a(S)} \, \Delta + \bigl(1-\chrb{z(S)>a(S)}\bigr)\,a(S)}{2} + \frac{a(S)}{2} \\
	&= a(S) + \frac{1}{2} \,\chrb{z(S) > a(S)} \bigl(\Delta - a(S)\bigr),
\end{align*}
where $\chr{}$ denotes the indicator function on binary events, \ie $\chr{E}=1$ if event $E$ takes place and $\chr{E}=0$ otherwise.
Depending on the value of $z(S)$, the right-hand side is minimized either for $a(S)=0$ or for $a(S)=z(S)$, and it becomes equal to $\chr{z(S)>0}\cdot\frac{\Delta}{2}$ when $a(S)=0$ and equal to
$z(S)$ when $a(S)=z(S)$. In summary,
\begin{equation}  \label{eq:expect}
	\E[X\cond S\setminus\{\imax\}=S]\geq\min\biggl\{\chrb{z(S)>0} \cdot \frac{\Delta}{2}, z(S)\biggr\} = \min
\biggl\{\frac{\Delta}{2}, z(S)\biggr\}.
\end{equation}
We can now lift the assumption that $S\neq\emptyset$ and $S\neq N\setminus\{\imax\}$. If $S=\emptyset$, then $z(S)=0$
and~\eqref{eq:expect} holds trivially. If $S=N\setminus\{\imax\}$, then $z(S)=\Delta$, and $i^*$ is in $N\setminus S$ and
therefore chosen by the $2$-partition mechanism with probability $1/2$. Thus $\E[X\cond A_1 \setminus\{\imax\}=S]\geq
\Delta/2=\min\{\Delta/2, z(S)\}$, and~\eqref{eq:expect} is again satisfied.

By construction of the $2$-partition mechanism, each vertex belongs to $A_1$ with probability $1/2$, so $z(S) =
\deg{A_1}{i^*}$ is distributed according to the binomial distribution with $\Delta$ trials and success probability $1/2$. We
thus have that
\begin{align*}
	\E[X] &= \sum_{S \subseteq N} \P[A_1 \setminus \{\imax\} = S] \cdot \E[X \cond A_1 \setminus \{\imax\} = S] \\
		&\geq \sum_{k=0}^{\Delta} \sum_{\substack{S\subseteq N\\z(S) = k}} \P[A_1 \setminus \{\imax\} = S] \cdot \min\biggl\{\frac{\Delta}{2},k\biggr\} \\
		&= \frac{1}{2^{\Delta}} \sum_{k=0}^{\Delta} \binom{\Delta}{k} \cdot \min\biggl\{\frac{\Delta}{2},k\biggr\}.
\end{align*}
This finally implies that 
\begin{equation*}
	\alpha_2(\Delta)\geq \frac{1}{\Delta 2^{\Delta}} \sum_{k=0}^{\Delta} \binom{\Delta}{k} \cdot
	\min\biggl\{\frac{\Delta}{2},k\biggr\}
\end{equation*}
as claimed.
\end{proof}

We now use \lemref{lem:alpha_2} to derive a closed form expression for $\alpha_2(\Delta)$.
\begin{theorem}  \label{thm:alpha_2}
	On any graph $G$ with maximum degree $\Delta=\Delta(G)$, the $2$-partition mechanism is $\alpha_2(\Delta)$-optimal, where
\begin{align*}
	\alpha_2(\Delta) = \begin{cases}
		\frac{1}{4} & \text{if $\Delta=1$,} \\  
		\frac{1}{2} - \frac{1}{2^{\Delta+2}} \binom{\Delta}{\Delta/2} &\text{if $\Delta\geq 2$ and even}, \\
		\alpha_2(\Delta-1) & \text{if $\Delta\geq 3$ and odd.}
	\end{cases}
\end{align*}
\end{theorem}
\begin{proof}
	Using \lemref{lem:alpha_2} it is straightforward to calculate that $\alpha_2(1)=1/4$.

	Next assume that $\Delta$ is strictly positive and even. Then, by \lemref{lem:alpha_2},
\begin{align*}
	\alpha_2(\Delta) &= \frac{1}{\Delta 2^{\Delta}}  \sum_{k=0}^{\frac{\Delta}{2}-1} \binom{\Delta}{k} \cdot k + \frac{1}{\Delta
2^{\Delta}} \binom{\Delta}{\Delta/2} \cdot \frac{\Delta}{2} + \frac{1}{\Delta 2^{\Delta}} \; \sum_{\mathrlap{k=\frac{\Delta}{2}+1}}^{\Delta} \; \binom{\Delta}{k} \cdot \frac{\Delta}{2}
\end{align*}
By symmetry of the binomial distribution with success probability $1/2$,
\begin{equation*}
	\frac{1}{2^{\Delta}} \sum_{k=0}^{\frac{\Delta}{2}-1} \binom{\Delta}{k} = \frac{1}{2^{\Delta}} \;
\sum_{\mathrlap{k=\frac{\Delta}{2}+1}}^{\Delta} \; \binom{\Delta}{k}
\end{equation*}
and thus
\begin{equation*}
	\frac{1}{2} \cdot \frac{1}{2^{\Delta}}
	\binom{\Delta}{\Delta/2} + \frac{1}{2^{\Delta}} \sum_{k=0}^{\frac{\Delta}{2}-1} \binom{\Delta}{k} = 
	\frac{1}{2} \cdot \frac{1}{2^{\Delta}} \binom{\Delta}{\Delta/2} + \frac{1}{2^{\Delta}}\; \sum_{\mathrlap{k=\frac{\Delta}{2}+1}}^{\Delta}\; \binom{\Delta}{k} =
	\frac{1}{2} .
\end{equation*}
Using the latter, we obtain
\begin{align*}
	\alpha_2(\Delta) &= \frac{1}{\Delta 2^{\Delta}} \sum_{k=0}^{\frac{\Delta}{2}-1} \binom{\Delta}{k} \cdot k + \frac{1}{2} \cdot
\frac{1}{\Delta 2^{\Delta}} \binom{\Delta}{\Delta/2} \cdot \frac{\Delta}{2} + \frac{1}{4} \\
	&= \frac{1}{2^{\Delta}} \sum_{k=1}^{\frac{\Delta}{2}-1} \frac{(\Delta-1)!}{(\Delta-k)!(k-1)!} + \frac{1}{2^{\Delta+2}} \binom{\Delta}{\Delta/2} + \frac{1}{4} \\
	&= \frac{1}{2^{\Delta}} \sum_{k=1}^{\frac{\Delta}{2}-1} \binom{\Delta-1}{k-1} + \frac{1}{2^{\Delta+2}} \binom{\Delta}{\Delta/2} + \frac{1}{4} \\
	&= \frac{1}{2} \cdot \frac{1}{2^{\Delta-1}} \sum_{k=0}^{\frac{\Delta}{2}-2} \binom{\Delta-1}{k} + \frac{1}{2^{\Delta+2}} \binom{\Delta}{\Delta/2} + \frac{1}{4} .
\end{align*}
Also by symmetry of the binomial distribution,
\begin{equation*}
	\frac{1}{2^{\Delta-1}} \sum_{k=0}^{\frac{\Delta}{2}-1} \binom{\Delta-1}{k} = \frac{1}{2} ,
\end{equation*}
and we obtain 
\begin{align*}
	\alpha_2(\Delta) &= \frac{1}{4} + \frac{1}{2} \Biggl(\frac{1}{2} - \frac{1}{2^{\Delta-1}}\binom{\Delta-1}{\Delta/2-1}\Biggr)
+ \frac{1}{2^{\Delta+2}} \binom{\Delta}{\Delta/2} \\
	&= \frac{1}{2} - \frac{1}{2^{\Delta}} \binom{\Delta-1}{\Delta/2-1} + \frac{1}{2^{\Delta+2}} \binom{\Delta}{\Delta/2}.
\end{align*}
Since $\Delta-1$ is odd, $\binom{\Delta-1}{\Delta/2}=\binom{\Delta-1}{\Delta/2-1}$ and thus
$\binom{\Delta}{\Delta/2}=\binom{\Delta-1}{\Delta/2}+\binom{\Delta-1}{\Delta/2-1}=2\binom{\Delta-1}{\Delta/2-1}$. We
conclude that
\begin{align*}
	\alpha_2(\Delta) &= \frac{1}{2} - \frac{1}{2^{\Delta+1}} \binom{\Delta}{\Delta/2} + \frac{1}{2^{\Delta+2}}
\binom{\Delta}{\Delta/2} \\
	&= \frac{1}{2} - \frac{1}{2^{\Delta+2}} \binom{\Delta}{\Delta/2}
\end{align*}
as claimed.

Finally assume that $\Delta\geq 3$ and odd. Then, by \lemref{lem:alpha_2},
\begin{align*}
	\alpha_2(\Delta) &= \frac{1}{\Delta 2^{\Delta}}  \sum_{k=0}^{\left\lfloor \frac{\Delta}{2} \right\rfloor} \binom{\Delta}{k} \cdot k +
\frac{1}{\Delta 2^{\Delta}} \; \sum_{\mathrlap{k=\left\lceil
\frac{\Delta}{2} \right\rceil}}^{\Delta} \; \binom{\Delta}{k} \cdot \frac{\Delta}{2} .
\end{align*}
By symmetry of the binomial distribution with success probability $1/2$,
\begin{equation*}
	\frac{1}{2^{\Delta}} \; \sum_{\mathrlap{k=\left\lceil \frac{\Delta}{2} \right\rceil}}^{\Delta} \; \binom{\Delta}{k} = \frac{1}{2} ,
\end{equation*}
and we obtain
\begin{align*}	
\alpha_2(\Delta) &= \frac{1}{\Delta 2^{\Delta}}  \sum_{k=0}^{\frac{\Delta-1}{2}} \binom{\Delta}{k} \cdot k + \frac{1}{4} \\
&= \frac{1}{2^{\Delta}} \sum_{k=1}^{\frac{\Delta-1}{2}} \frac{(\Delta-1)!}{(\Delta-k)!(k-1)!} + \frac{1}{4} \\
&= \frac{1}{2^{\Delta}} \sum_{k=1}^{\frac{\Delta-1}{2}}\binom{\Delta-1}{k-1} + \frac{1}{4} \\
&= \frac{1}{2} \cdot \frac{1}{2^{\Delta-1}} \; \sum_{k=0}^{\mathclap{\frac{\Delta-1}{2}-1}} \binom{\Delta-1}{k} + \frac{1}{4}.
\end{align*}
Since $\Delta-1$ is even, and again by symmetry of the binomial distribution, \begin{equation*}
	\frac{1}{2^{\Delta-1}} \; \sum_{k=0}^{\mathclap{\frac{\Delta-1}{2}-1}} \binom{\Delta-1}{k}  + \frac{1}{2} \cdot \frac{1}{2^{\Delta-1}} \binom{\Delta-1}{(\Delta-1)/2} = \frac{1}{2} .
\end{equation*}
We conclude that
\begin{equation*}
	\alpha_2(\Delta) 
	= \frac{1}{4} + \frac{1}{2} \Biggl(\frac{1}{2} -
\frac{1}{2^{\Delta}}\binom{\Delta-1}{(\Delta-1)/2}\Biggr) 
	= \alpha_2(\Delta-1)
\end{equation*}
as claimed.
\end{proof}

Given this closed-form expression, it is not difficult to show that $\alpha_2(\Delta)$ is non-decreasing in $\Delta$.
\begin{corollary}
	For every $\Delta\in\N$, $\alpha_2(\Delta+1)\geq \alpha_2(\Delta)$ and $\alpha_2(\Delta+2)>\alpha_2(\Delta)$.
\end{corollary}
\begin{proof}
	Since $\alpha_2(\Delta)=\alpha_2(\Delta-1)$ for odd $\Delta\geq 3$ by \thmref{thm:alpha_2}, it suffices to show that $\alpha_2(\Delta)>\alpha_2(\Delta-2)$ for even $\Delta\geq 4$. To this end, note that
\begin{align*}
	\alpha_2(\Delta) &= \frac{1}{2} - \frac{1}{2^{\Delta+2}} \binom{\Delta}{\Delta/2} .
\end{align*}
Using three times that $\binom{\Delta}{k}=\binom{\Delta-1}{k-1}+\binom{\Delta-1}{k}$, we obtain
\begin{align*}
	\alpha_2(\Delta) &= \frac{1}{2} - \frac{1}{2^{\Delta+2}} \Biggl(\binom{\Delta-2}{\Delta/2-2} + 2\binom{\Delta-2}{\Delta/2-1}
+ \binom{\Delta-2}{\Delta/2}\Biggr) ,
\end{align*}
and since $\binom{\Delta-2}{k}$ is maximized for $k=\Delta/2-1$, 
\begin{align*}
	\alpha_2(\Delta) &> \frac{1}{2} - \frac{1}{2^{\Delta}} \binom{\Delta-2}{\Delta/2-1} = \alpha_2(\Delta-2). \tag*{\raisebox{-.5\baselineskip}{\qedhere}}
\end{align*}
\end{proof}

This result implies that a graph with a single edge is in fact the unique worst case for the $2$-partition mechanism, and it also yields the first non-trivial lower bound for settings without abstentions. In the absence of abstentions, one of two conditions is always satisfied: either every vertex has indegree exactly one, in which case every mechanism including $2$-partition is optimal, or $\Delta\geq 2$ and $2$-partition is at least $\alpha_2(\Delta)$-optimal. We will return to this special case, and show a better bound, in \secref{sec:no_abstentions}.

\section{The $k$-Partition Mechanism}

What perhaps is most interesting about the above analysis of the $2$-partition mechanism is that the technique we have used to analyse it can in principle also be applied to a partition of the vertices into more than two sets. Indeed, in this section, we propose a generalization of the $2$-partition mechanism to a larger number of sets and then generalize the analysis technique to the new mechanism.

For a fixed $k\geq 2$, the new mechanism first assigns each vertex $i\in N$ independently and uniformly at random to one of $k$ sets $A_1,\dots,A_k$, such that $\bigcup_{i=1,\dots,k}A_i=N$, $A_i\cap A_j=\emptyset$ for all $i,j\in\{1,\dots,k\}$ with $i\neq j$, and $\P[i\in A_i]=1/k$ for all $i\in
\{1,\dots,k\}$. The mechanism then proceeds in $k-1$ iterations numbered from $2$ to $k$, during which it maintains and updates a candidate vertex that is finally selected after iteration $k$. The candidate is updated if the maximum indegree among vertices in $A_j$ from vertices in $A_{<j}=\bigcup_{i=1}^{j-1}A_i$ other than the candidate is at least that of the candidate at the time it became the candidate. In that case, the new candidate is chosen uniformly at random from the set of vertices in $A_j$ with maximum indegree from vertices in $A_{<j}=\bigcup_{i=1}^{j-1}A_i$, now including the previous candidate. The mechanism is clearly impartial, because it only takes into account the outgoing edges of vertices that can no longer be selected. The fact the any outgoing edges of the previous candidate are taken into account when selecting the new candidate is somewhat subtle, but it turns out to be crucial for our results.
A formal description of the mechanism is given in \figref{alg:k_partition}.
\begin{algorithm}[tb]
	\KwIn{Graph $G = (N,E)$}
	\KwOut{Vertex $i \in N$}
	Assign each $i\in N$ independently and uniformly at random to one of $k$ sets $A_1,\dots,A_k$\; \label{line:k_partition}
	Choose $\imax\in A_1$ arbitrarily; set $\dmax:=0$\;
	\For{$j=2,\dots,k$}{
    \If{$\max_{i\in A_j} \deg{A_{<j}\setminus\{\imax\}}{i}\geq \dmax$}{
        Choose $\imax\in\arg\max_{i'\in A_j}\deg{A_{<j}}{i'}$ uniformly at random; set $\dmax:=\deg{A_{<j}}{\imax}$\; \label{line:choose}
    }
	}
	Return $\imax$\;
\caption{The $k$-partition mechanism}
\label{alg:k_partition}
\end{algorithm}

Now consider a graph $G=(N,E)\in\G$ and a vertex $\imax\in N$ with degree $\Delta=\Delta(G)$. Fix $k\in\N$, and let $X$ be the degree of the vertex selected from $G$ by the $k$-partition mechanism. Note that $X$ is a random variable subject to the internal randomness of the mechanism, and that we are interested in its expected value $\E[X]$.

We need some notation. For a subset $N'\subseteq N$ of the vertices, let $\Part_k(N')$ denote the set of all partitions
$\S=(S_1,\dots,S_k)$ of $N'$ into $k$ (possibly empty) sets $S_1,\dots,S_k$, \ie $\Part_k(N')=\{(S_1,\dots,S_k)\midd S_j
\subseteq N' \text{ for } j=1,\dots,k,\, \bigcup_{j=1}^k S_j=N',\, S_i\cap S_j=\emptyset \text{ for } i,j=1,\dots,k \text{
with } i\neq j\}$. For a partition $\S=(S_1,\dots,S_k)$ and $j\in\{1,\dots,k\}$, let $S_{<j} = \bigcup_{i=1}^{j-1}S_i$. 
For a partition $\S=(S_1,\dots,S_k)\in\Part(N)$ and $i\in N$, we slightly abuse notation and write
$\S\setminus\{i\}=(S_1\setminus\{i\},\dots,S_k\setminus\{i\})$ for the partition obtained from $\S$ by removing $i$.

Let $\A$ be the partition chosen in \lineref{line:k_partition} of the $k$-Partition mechanism in \figref{alg:k_partition}. The following lemma bounds the expected value of $X$ given that $\A=\S$ for some given partition $\S\in \Part_k(N)$.
\begin{lemma}  \label{lem:fixed_partition}
	Consider a graph $G=(N,E)$ and a vertex $\imax$ with degree $\Delta=\Delta(G)$. Let $\S=(S_1,\dots,S_k)\in\Part_k(N)$, and let $j^*\in\{1,\dots,k\}$ such that $\imax\in S_{j^*}$. Then,
\[
	\E[X \cond \A = \S] \geq a + \chr{z > a}\cdot \bigl(\Delta-a\bigr) ,
\]
where $a=\max_{j=2,\dots,k}\max_{i\in S_j\setminus\{\imax\}} \deg{S_{<j}}{i}$ and $z=\deg{S_{<j^*}}{\imax}$.
\end{lemma}
\begin{proof}
	For $j=2,\dots,k$, let $\imax(j)$ and $\dmax(j)$ denote the values of $\imax$ and $\dmax$ after iteration~$j$ of the mechanism. We show by induction on $j$ that for all $j=2,\dots,k$, $\dmax(j)=\max_{m=2,\dots,j}\max_{i\in S_m}\deg{S_{<m}}{i}$.

	First consider the case where $j=2$. If $S_2=\emptyset$, there is nothing to show. Otherwise, the mechanism chooses a vertex $\imax(2)$ with $\deg{S_{<2}}{\imax(2)}=\deg{S_{1}}{\imax(2)}=\max_{i\in S_2}\deg{S_{1}}{i}$, and thus
\begin{equation*}
	\dmax(2) = \deg{S_1}{\imax(2)} = \max_{i\in S_2} \deg{S_1}{i} .
\end{equation*}

Now suppose that $\dmax(j-1)=\max_{m=2,\dots,j-1}\max_{i\in S_m}\deg{S_{<m}}{i}$ for some $j\in\{3,\dots,k\}$. If $S_j=\emptyset$, there again is nothing to show. Otherwise we consider iteration $j$ of the mechanism and distinguish two cases.

If $\max_{i\in S_j}\deg{S_{<j}\setminus\{\imax(j-1)\}}{i}\geq\dmax(j-1)$, then $\imax(j)\in\arg\max_{i\in S_j}\deg{S_{<j}}{i}$ and
\begin{equation*}
	\dmax(j) = \deg{S_{<j}}{\imax(j)} = \max_{i\in S_j}\deg{S_{<j}}{i} .
\end{equation*}
Furthermore,
\begin{equation*}
	\dmax(j) = \degb{S_{<j}}{\imax(j)} \geq \degb{S_{<j}\setminus \{\imax(j-1)\}}{\imax(j)}\geq\dmax(j-1) =	\max\nolimits_{m=2,\dots,j-1}\max\nolimits_{i\in S_m}\deg{S_{<m}}{i} ,
\end{equation*}
where the last equality holds by the induction hypothesis. In summary,
\begin{equation*}
	\dmax(j) = \max\nolimits_{m=2,\dots,j}\max\nolimits_{i\in S_m}\deg{S_{<m}}{i}.
\end{equation*}

If, on the other hand, $\max_{i\in S_j} \deg{S_{<j}\setminus\{\imax(j-1)\}}{i} < \dmax(j-1)$, then $\imax(j)=\imax(j-1)$ and 
\begin{equation*}
	\dmax(j) = \dmax(j-1) 
	\geq \max_{i\in S_j} \deg{S_{<j}\setminus\{\imax(j-1)\}}{i} + 1
	\geq \max_{i\in S_j} \deg{S_{<j}}{i} ,
\end{equation*}
where the first inequality holds because degrees are integral, and the second inequality because there can be at most one edge from $\imax(j-1)$ to $\imax(j)$.
Furthermore,
\begin{equation*}
	\dmax(j) = \dmax(j-1) = \max\nolimits_{m=2,\dots,j-1}\max\nolimits_{i\in S_m} \deg{S_{<m}}{i} 
\end{equation*}
where the second equality holds by the induction hypothesis, so again
\begin{equation*}
	\dmax(j) = \max\nolimits_{m=2,\dots,j}\max\nolimits_{i\in S_m} \deg{S_{<m}}{i} .
\end{equation*}

Since the mechanism returns $\imax(k)$,
\begin{align*}
	\E[X\cond \A=\S] = \deg{}{\imax(k)} = \dmax(k) 
	&= \max\nolimits_{j=2,\dots,k} \max\nolimits_{i\in S_j} \deg{S_{<j}}{i} \\
	&\geq \max\nolimits_{j=2,\dots,k} \max\nolimits_{i\in S_j} \deg{S_{<j} \setminus \{\imax\}}{i} = a .
\end{align*}
To complete the proof, assume that
\begin{equation*}
	z = \deg{S_{<j^*}}{\imax} > 
	\max\nolimits_{j=2,\dots,k} \max\nolimits_{i\in S_j\setminus\{\imax\}} \deg{S_{<j}}{i} = a .
\end{equation*}
Then
\begin{equation*}
	\deg{S_{<j^*}}{i^*} > \max\nolimits_{m=2,\dots,j^*-1} \max\nolimits_{i\in S_m}\deg{S_{<m}}{i} = \dmax(j^*-1)
\end{equation*}
and
\begin{equation*}
	\deg{S_{<j^*}}{i^*} > \max\nolimits_{m=j^*+1,\dots,k} \max\nolimits_{i \in S_m} \deg{S_{<m}}{i} ,
\end{equation*}
so $\imax(j)=\imax$ for $j=j^*,\dots,k$. The mechanism thus selects $\imax$, and $\E[X \cond \A=\S, z>a]=\Delta$.
\end{proof}

As in our analysis of the $2$-partition mechanism, we now proceed to bound the expected value of~$X$ given that a partition is fixed for all vertices except $\imax$, and $\imax$ is then allocated uniformly at random to one of the $k$ sets.
\begin{lemma}  \label{lem:fixed_super_partition}
	Consider a graph $G=(N,E)$ and a vertex $\imax$ with degree $\Delta=\Delta(G)$. Let $\S=(S_1,\dots,S_k)\in\Part_k(N\setminus\{\imax\})$. For $j=1,\dots,k$, let $z_j=\deg{S_{<j}}{\imax}$. Then,
\[
	\E[X \cond \A\setminus\{\imax\}=\S] \geq \min\nolimits_{j=1,\dots,k} \biggl\{ z_j + \frac{k-j}{k} (\Delta-z_j) \biggr\} .
\]
\end{lemma}
\begin{proof}
	There are exactly $k$ partitions $\S'\in\Part(N)$ such that $\S'\setminus\{i^*\}=\S$, and each of them occurs with probability $1/k$, so
\begin{equation*}
\E[X \cond \A \setminus \{i^*\} = \S] = \frac{1}{k} \sum_{m=1}^k \E[X \cond \A \setminus \{i^*\} = \S, i^* \in A_m] .
\end{equation*}
By \lemref{lem:fixed_partition},
\begin{equation*}
	\E[X \cond \A \setminus \{i^*\} = \S] \geq a(\S) + \frac{1}{k} \sum_{m=1}^k \chr{z_m > a(\S)} \cdot (\Delta-a(\S)), 
\end{equation*}
where $a(\S) = \max_{j=2,\dots,k} \max_{i \in S_j} \deg{S_{<j}}{i}$. Note that the right-hand side is minimized when $a(\S)=z_m$ for some $m\in\{1,\dots,k\}$, so
\begin{align*}
	\E[X \cond \A\setminus\{i^*\}=\S] &\geq \min\nolimits_{j=1,\dots,k} \biggl\{ z_j + \frac{1}{k} \sum_{m=1}^k \chr{z_m>z_j} \cdot (\Delta-z_j) \biggr\} \\
	&= \min\nolimits_{j=1,\dots,k} \biggl\{ z_j + \frac{1}{k} \; \sum_{\mathrlap{m=j+1}}^k \chr{z_m>z_j} \cdot (\Delta-z_j) \biggr\}.
\end{align*}
Now observe that whenever $z_j=z_{j+1}$ for some $j=2,\dots,k$, the respective terms in the minimization are equal as well. This implies that the minimum will always be attained for some $j$ with $z_{j-1}<z_j$.
By setting $z_0=-1$ and simplifying,
\begin{align*}
	\E[X \cond \A \setminus \{i^*\} = \S] 
	&\geq \min\nolimits_{\substack{j=1,\dots,k\\z_{j-1}<z_{j}}} \biggl\{ z_j + \frac{1}{k} \sum_{m=j+1}^k (\Delta-z_j) \biggr\} \\
	&= \min\nolimits_{\substack{j=1,\dots,k\\z_{j-1}<z_{j}}} \biggl\{ z_j + \frac{k-j}{k} (\Delta-z_j) \biggr\} .
\end{align*}
Since $z_j+\frac{k-j}{k}(\Delta-z_j)$ attains its minimum for some $j$ with $z_{j-1}<z_{j}$, we can drop the condition that $z_{j-1}<z_{j}$ and obtain
\begin{equation*}
	\E[X \cond \A\setminus\{i^*\}=\S] \geq \min\nolimits_{j=1,\dots,k} \biggl\{ z_j + \frac{k-j}{k}(\Delta-z_j) \biggr\} 
\end{equation*}
as claimed.
\end{proof}

To obtain a bound on $\E[X]$, we will now average the expression obtained in \lemref{lem:fixed_super_partition} over the distribution on partitions of $N$. Like the former, the bound we obtain does not depend on the actual partitions, but only on the indegree $\deg{S_j}{\imax}$ of $\imax$ from each set $S_j$ in the partition. For $\Delta,k\in\N$, let $\part_k(\Delta)=\{\vec{v}\in\N^k\midd \sum_{i=1}^k v_i=\Delta\}$. For $\vec{v}\in\part_k(\Delta)$, let $\binom{\Delta}{\vec{v}}=\frac{\Delta!}{v_1!\cdots v_k!}$ be the number of partitions of a set with $\Delta$ elements into~$k$ sets of sizes $v_1,\dots,v_k$. We then have the following result.
\begin{lemma}  \label{lem:alpha_k}
	On any graph $G$ with maximum degree $\Delta=\Delta(G)$, the $k$-partition mechanism is $\alpha_k(\Delta)$-optimal, where
\[
	\alpha_k(\Delta) = \frac{1}{\Delta k^{\Delta}} \sum\nolimits_{\vec{v}\in\part_k(\Delta)} \binom{\Delta}{\vec{v}} \min\nolimits_{j=1,\dots,k} \Biggl\{\frac{k-j}{k} \sum_{\ell= 1}^{k} v_{\ell} + \frac{j}{k} \sum_{\ell=1}^{j-1} v_{\ell} \Biggr\}.
\]
\end{lemma}
\begin{proof}
	Consider a vertex $\imax$ with degree $\Delta$, and note that
\begin{align*}
	\E[X] = \sum\nolimits_{\S \in \Part_k(N \setminus \{i^*\})} \P[\A \setminus \{i^*\} = \S] \cdot \E[X \cond \A \setminus \{i^*\} = \S] .
\end{align*}
For $\S\in\Part_k(N\setminus\{i^*\})$, let $z_j(\S)=\deg{S_{<j}}{i^*}$. Then, by Lemma~\ref{lem:fixed_super_partition},
\begin{align*}
	\E[X] \geq \sum\nolimits_{\S \in \Part_k(N \setminus \{i^*\})} \P[\A \setminus \{i^*\} = \S] \cdot \min\nolimits_{j = 1,\dots,k} 
\biggl\{z_j(\S) + \frac{k-j}{k} (\Delta-z_j(\S))\biggr\}.
\end{align*}
For $\S\in\Part_k(N\setminus\{i^*\})$, let $v_j(\S)=\deg{S_{j}}{i^*}$. Then, $z_j(\S)=\sum_{m=1}^{j-1} v_m(\S)$, and
\begin{align*}
	\E[X] &\geq \sum\nolimits_{\S \in \Part_k(N \setminus \{i^*\})} \P[\A \setminus \{i^*\} = \S] \cdot \min\nolimits_{j=1,\dots,k}
	\Biggl\{\sum_{m=1}^{j-1} v_m(\S) + \frac{k-j}{k} \sum_{m=j}^{k} v_m(\S)\Biggr\}\\
	&= \sum\nolimits_{\S \in \Part_k(N \setminus \{i^*\})} \P[\A \setminus \{i^*\} = \S] \cdot \min\nolimits_{j=1,\dots,k} \Biggl\{\frac{k-j}{k}
\sum_{m=1}^{k} v_m(\S) + \frac{j}{k} \sum_{m=1}^{j-1} v_m(\S)\Biggr\} .
\end{align*}
Since $(z_1(\S),\dots,z_k(\S))$ follows a multinomial distribution with $\Delta$ trials and success probability $1/k$ for each category, we have that \begin{align*}
	\E[X] \geq  \frac{1}{k^{\Delta}} \sum\nolimits_{\vec{v}\in\part_k(\Delta)} \binom{\Delta}{\vec{v}} \min\nolimits_{j=1,\dots,k} \Biggl\{\frac{k-j}{k} \sum_{\ell= 1}^{k} v_{\ell} + \frac{j}{k} \sum_{\ell=1}^{j-1} v_{\ell} \Biggr\},
\end{align*}
and the claim follows.
\end{proof}

In the case of the $2$-partition mechanism, we obtained a lower bound on the degree of optimality by deriving a closed-form expression for $\alpha(\Delta)=\alpha_2(\Delta)$ that turned out to be monotonically non-decreasing in $\Delta$. While the complexity of $\alpha_k$ prevents us from taking the same route for $k>2$, monotonicity turns out to hold for any value of~$k$.
\begin{theorem}  \label{thm:alpha_k_mono}
	For any $k\geq 2$, $\alpha_k(\Delta)$ is non-decreasing in $\Delta$.
\end{theorem}
\begin{proof}
	We can reformulate the bound of \lemref{lem:alpha_k} to obtain that
\begin{align*}
	\alpha_k(\Delta) &= \frac{1}{\Delta k^{\Delta}} \sum\nolimits_{\vec{v}\in\part_k(\Delta)} \binom{\Delta}{\vec{v}} \min\nolimits_{j=1,\dots,k} \Biggl\{\frac{k-j}{k} \sum_{\ell= 1}^{k} v_{\ell} + \frac{j}{k} \sum_{\ell=1}^{j-1} v_{\ell} \Biggr\} \\
	&= \frac{1}{\Delta k^{\Delta}} \sum\nolimits_{\vec{v}\in\part_k(\Delta)} \binom{\Delta}{\vec{v}} \min\nolimits_{j=1,\dots,k} \langle \vec{v}, \vec{w}^{j,k}\rangle ,
\end{align*}
where $\vec{w}^{j,k}\in\Q^k$ 
with 
$w^{j,k}_i=1$ if $i<j$ and $w^{j,k}_i=(k-j)/k$ otherwise.

Instead of summing over all vectors $\vec{v}\in\part_k(\Delta)$, we may instead sum over all vectors $\vec{v}'\in\part_k(\Delta+1)$ and decrease one of the non-zero entries $v'_i$ by $1$. Thus,
\begin{align*}
\alpha_k(\Delta) &= \frac{1}{\Delta k^{\Delta+1}} \sum\nolimits_{\vec{v}\part_k(\Delta+1)} \binom{\Delta+1}{\vec{v}} \frac{1}{\Delta+1} \sum_{i=1}^k v_i \min\nolimits_{j=1,\dots,k}  \langle \vec{v} - \vec{e}^{i,k}, \vec{w}^{j,k} \rangle, 
\end{align*}
where $\vec{e}^{i,k}$ is the $i$th unit vector in $k$ dimensions, \ie $e^{i,k}_\ell=1$ if $\ell=i$ and $e^{i,k}_\ell=0$ otherwise.

If we exchange the order of the summation over $i$ and the minimization over $j$, the value of the expression can only increase, so
{\allowdisplaybreaks
\begin{align*}
	\alpha_k(\Delta) & \leq \frac{1}{\Delta k^{\Delta+1}} \sum\nolimits_{\vec{v}\in\part(\Delta+1)} \binom{\Delta+1}{\vec{v}} \frac{1}{\Delta+1} \min\nolimits_{j=1,\dots,k}  \sum_{i=1}^k v_i  \langle \vec{v} - \vec{e}^{i,k}, \vec{w}^{j,k} \rangle \\
	&= \frac{1}{(\Delta+1) k^{\Delta+1}} \sum\nolimits_{\vec{v}\in\part_k(\Delta+1)} \binom{\Delta+1}{\vec{v}} \frac{1}{\Delta} \min\nolimits_{j=1,\dots,k}  \sum_{i=1}^k \bigl(v_i  \langle \vec{v}, \vec{w}^{j,k} \rangle - v_i\, w^{j,k}_i\bigr)\\
	&= \frac{1}{(\Delta+1) k^{\Delta+1}} \sum\nolimits_{\vec{v}\in\part_k(\Delta+1)} \binom{\Delta+1}{\vec{v}} \frac{1}{\Delta} \min\nolimits_{j=1,\dots,k}  \Bigl( \langle \vec{v}, \vec{w}^{j,k}\rangle \sum_{i=1}^k v_i  - \langle \vec{v}, \vec{w}^{j,k} \rangle \Bigr) \\
	&= \frac{1}{(\Delta+1) k^{\Delta+1}} \sum\nolimits_{\vec{v}\in\part_k(\Delta+1)} \binom{\Delta+1}{\vec{v}} \min\nolimits_{j=1,\dots,k}  \langle \vec{v}, \vec{w}^{j,k}\rangle \\[1ex]
	&= \alpha_k(\Delta+1)  \qedhere
\end{align*}}
\end{proof}

Monotonicity of $\alpha_k$ allows us to obtain a lower bound on the degree of optimality of the $k$-partition mechanism by bounding $\alpha_k(1)$ from below. The following is our main result.
\begin{theorem}  \label{thm:k-partition}
	The $k$-partition mechanism for $k\geq 2$ is $\frac{k-1}{2k}$-optimal.
\end{theorem}
\begin{proof} 
	In light of \thmref{thm:alpha_k_mono}, it suffices to show that $\alpha_k(1) \geq \frac{k-1}{2k}$ for every $k\geq 2$.

	By \lemref{lem:alpha_k},
\begin{equation*}
	\alpha_k(1) = \frac{1}{k} \sum\nolimits_{\vec{v}\in\part_k(1)} \min\nolimits_{j=1,\dots,k} \Biggl\{ \frac{k-j}{k} \sum_{\ell=1}^k v_{\ell} + \frac{j}{k} \sum_{\ell=1}^{j-1} v_{\ell}\Biggr\}
\end{equation*}
Taking the pointwise minimum,
\begin{equation*}
	\alpha_1(k) \geq \frac{1}{k} \sum\nolimits_{\vec{v}\in\part_k(1)} \Big\langle \vec{v}, \Bigl( \frac{k-1}{k},\frac{k-2}{k},\dots,\frac{1}{k},0\Bigr)\Big\rangle
\end{equation*}
In the sum, every unit vector occurs exactly once, and thus
\begin{align*}
	\alpha_k(1) \geq \frac{1}{k} \sum_{i=1}^{k} \frac{k-i}{k} = \frac{1}{k^2} \sum_{i=0}^{k-1} i = \frac{k (k-1)}{2k^2} = \frac{k-1}{2k}.  \tag*{\raisebox{-.5\baselineskip}{\qedhere}}
\end{align*}
\end{proof}

\section{The Permutation Mechanism}
\label{sec:permutation}

We have started from the simple result that no impartial selection mechanism can be more than $1/2$-optimal, and in the previous section identified a class of mechanisms parameterized by $k\in\N$ that attains this bound in the limit as $k$ goes to infinity. It turns out that the bound can also be attained exactly, by a limiting mechanism for the above class.
This mechanism, which we call the permutation mechanism, considers the vertices one by one according to a random permutation $\pi=(\pi_1,\dots,\pi_n)$ and in each step campares the current vertex $\pi_j$ to a single candidate vertex $\pi_{\ell}$ with $\ell<j$. In determining the degree of the candidate vertex $\pi_\ell$ it takes into account the outgoing edges of vertices $\pi_1,\dots,\pi_{\ell-1}$. For the degree of the current vertex $\pi_j$ it takes into account the outgoing edges of vertices $\pi_1,\dots,\pi_{j-1}$, except $\pi_{\ell}$. If the latter is greater than or equal than the former, $\pi_j$ becomes the new candidate vertex, and the candidate vertex after the final step is the one selected by the mechanism. Again it is easy to see that this mechanism is impartial, because it only takes into account the outgoing edges of vertices that can no longer be selected.
\begin{algorithm}[tb]
\KwIn{Graph $G = (N,E)$}
\KwOut{Vertex $i \in N$}
Choose a permutation $\bigl(\pi_1,\dots,\pi_{|N|}\bigr)$ of $N$ uniformly at random\;
\label{line:permutation}
Set $\imax := \pi_1$, $\dmax := 0$\;
\For{$j=2,\dots,|N|$}{
    \If{$\deg{\pi_{<j} \setminus \{\imax\}}{\pi_j} \geq \dmax$}{
        Set $\imax := \pi_j$, $\dmax := \deg{\pi_{<j}}{\pi_j}$\;
    }
}
return $\imax$\;
\caption{The permutation mechanism}
\label{alg:permutation}
\end{algorithm}
A formal description of the mechanism is given in \figref{alg:permutation}, and we obtain the following result.
\begin{theorem}
\label{thm:permutation}
The permutation mechanism is $1/2$-optimal.
\end{theorem}
\begin{proof}
	Assume for the sake of contradiction that there exists a graph $G=(N,E)$ such that the permutation mechanism is strictly less than $1/2$-optimal on $G$. Let $n=|N|$ and $\Delta=\Delta(G)$, and denote by $X$ and by $X_k$ for $k\geq 2$ the degree of the vertex respectively selected from $G$ by the permutation and $k$-partition mechanisms. Note that $X$ and $X_k$ are random variables subject to the internal randomness of the respective mechanism. Finally let $\alpha=\E[X]/\Delta$, and note that $\alpha<1/2$ by assumption.
	
	The outcomes of the permutation mechanism and the $k$-partition mechanism agree under the condition that the partition $(A_1,\dots,A_k)$ chosen by the latter satisfies $|A_i|\leq$ for $i=1,\dots,k$, so
	\begin{equation*}
		\E[X] \geq \P\bigl[\text{$|A_i|\leq 1$ for $i=1\dots,k$}\bigr] \cdot \E[X_k] .
	\end{equation*}
	For any $k\geq n$,
	\begin{align*}
	\P\bigl[\text{$|A_i|\leq 1$ for $i=1\dots,k$}\bigr] 
	&= \frac{k \cdot (k-1)  \cdot\, \dots \,\cdot (k-n+1)}{k^n} \geq \frac{(k-n)^{\mathrlap{n}}}{k^n} ,
	\end{align*}
	and thus
	\begin{equation*}
		\E[X] \geq \frac{k-n+1}{k}\; \E[X_k] \geq \frac{(k-n)^{\mathrlap{n}}}{k^n}\;\; \frac{k-1}{2k}\; \Delta ,
	\end{equation*}
	where the second inequality follows from \thmref{thm:k-partition}.
	For any fixed $n$, 
	\begin{equation*}
		\lim_{\mathclap{k\rightarrow\infty}} \left(\frac{(k-n)^{\mathrlap{n}}}{k^n}\;\; \frac{k-1}{2k}\right) = \frac{1}{2} > \alpha ,
	\end{equation*}
	and we can choose $k$ such that
	\begin{equation*}
		\frac{(k-n)^{\mathrlap{n}}}{k^n}\;\; \frac{k-1}{2k} > \alpha .
	\end{equation*}
	Therefore, $\E[X]>\alpha\Delta$, a contradiction. 
\end{proof}

A potential downside of the permutation mechanism is that it considers agents one by one and therefore cannot process nominations anonymously. This may be of concern in situations where agents do not want their opinion regarding other agents to be publicly known. In the $k$-partition mechanism for some fixed value of $k$, on the other hand, the nominations submitted by agents in block $A_j$ of the partition can be processed simultaneously and thus with partial anonymity. It is an interesting question whether this tradeoff between anonymity and approximate optimality is intrinsic to the problem, or whether there exits a different mechanism that achieves the same degree of optimality as the permutation mechanism but a higher degree of anonymity.

\section{No Abstentions}
\label{sec:no_abstentions}

Let us finally consider the interesting special case of graphs in which every vertex has outdegree at least~$1$. This case corresponds to settings without abstentions and in particular includes the setting of \citet{HoMo13a}, where every agent submits exactly one nomination. 
In \secref{sec:2_partition} we obtained an improved lower bound of $3/8$ instead of $1/4$ for the special case using the simple observation that the $2$-partition mechanism is optimal on graphs with maximum degree~$1$ and we can therefore focus on graphs with maximum degree at most~$2$. The same argument can be applied to the $k$-partition mechanism as well.

Let $\G_n^{+}= \{(N,E)\in\G_n\midd \min_{i\in N}\degp{i,(N,E)}\geq 1\}$, where $\degp{i,(N,E)}=|\{(i,j)\in E\midd j\in N\}\}|$, and $\G^+=\bigcup_{n\in\N}\G_n^+$.
The following lemma provides an alternative expression for $\alpha_k(2)$, which will subsequently be used to derive an improved bound.
\begin{lemma}  \label{lem:delta_2}
	On any graph $G$ with maximum degree $\Delta(G)=2$, the $k$-partition mechanism for $k\geq 2$ is $\alpha_k(2)$-optimal, where
\begin{equation*}
	\alpha_k(2) = 1 - \frac{1}{k^3} \sum\nolimits_{x_1,x_2\in\{1,\dots,k\}} \max \biggl\{ \frac{\max\{x_1,x_2\}}{2}, \min\{x_1,x_2\} \biggr\} .
\end{equation*}
\end{lemma}
\begin{proof}
	Let $k\geq 2$. Consider a graph $G=(N,E)$ with $\Delta(G)=2$, and note that there must exist $i_1,i_2,\imax\in N$ such that $i_1\neq i_2$ and $\{(i_1,\imax),(i_2,\imax)\}\subseteq E$. Consider the partition $\vec{A}=(A_1,\dots,A_k)$ chosen by the $k$-partition mechanism and let $x_1,x_2,y\in\{1,\dots,k\}$ such that $i_1\in A_{x_1}$, $i_2\in A_{x_2}$, and $\imax\in A_{y}$. Note that $x_1$, $x_2$, and $y$ are independent random variables distributed uniformly on $\{1,\dots,k\}$. Denote $\xmax=\max\{x_1,x_2\}$, and $\xmin=\min\{x_1,x_2\}$. Then, by \lemref{lem:fixed_super_partition},
\begin{align*} 
	\alpha_k(2) &\geq \frac{1}{2} \cdot \frac{1}{k^3} \sum\nolimits_{x_1,x_2,y\in\{1,\dots,k\}} \Bigr(2 \chrb{y >
\xmax} + \chrb{y \leq \xmax} \cdot \chrb{\xmax \leq 2(\xmax - \xmin)} + {} \\[-1ex]
	& \hspace*{6.575cm}
	2\chrb{\xmin \leq y \leq \xmax} \cdot \chrb{2(\xmax - \xmin) < \xmax} \! \Bigr) \\
	&= \frac{1}{k^2} \sum\nolimits_{x_1,x_2\in\{1,\dots,k\}} \biggl( \frac{k-\xmax}{k} + \frac{1}{2} \cdot \frac{\xmax}{k} \cdot
\chrb{2\xmin \leq \xmax} + \frac{\xmax-\xmin}{k} \cdot \chrb{2\xmin > \xmax} \! \biggr) \\
	&= \frac{1}{k^2} \sum\nolimits_{x_1,x_2\in\{1,\dots,k\}} \biggl(1 - \frac{\xmax}{2k} \cdot \chrb{2\xmin \leq \xmax} - \frac{\xmin}{k} \cdot \chrb{2\xmin > \xmax} \! \biggr) \\
	&= 1 - \frac{1}{k^3} \sum\nolimits_{x_1,x_2\in\{1,\dots,k\}} \max\biggl\{ \frac{\xmax}{2}, \xmin \biggr\} . \tag*{\raisebox{-.5\baselineskip}{\qedhere}}
\end{align*}
\end{proof}

We are now in a position to prove the main result of this section.
\begin{theorem}
	The permutation mechanism is $7/12$-optimal on $\G^{+}$.
\end{theorem}
\begin{proof}
	Let $G=(N,E)\in\G^{+}$. If $\Delta(G)=1$, then $\deg{}{i,G}=1$ for all $i\in N$ and every mechanism, including the permutation mechanism, is optimal on $G$. If $\Delta(G)\geq 2$, then by \lemref{lem:alpha_k}, \thmref{thm:alpha_k_mono}, and \lemref{lem:delta_2}, the $k$-partition mechanism is $\alpha_k(2)$-optimal on $G$, where $\alpha_k(2) = 1 - \beta_k/k^3$ with 
\begin{equation*}
	\beta_k = \sum\nolimits_{x_1,x_2\in\{1,\dots,k\}} \max \biggl\{ \frac{\max\{x_1,x_2\}}{2}, \min\{x_1,x_2\} \biggr\} .
\end{equation*}
Now, 
\begin{align*}
	\beta_k ={}& \sum\nolimits_{x_1,x_2\in \{1,\dots,k-1\}} \max \biggl\{ \frac{\max\{x_1,x_2\}}{2}, \min\{x_1,x_2\} \biggr\} + {} \\
	& \sum\nolimits_{x_2\in\{1,\dots,k-1\}} \max \biggl\{ \frac{\max\{k,x_2\}}{2}, \min\{k,x_2\} \biggr\} + {} \\
	& \sum\nolimits_{x_1\in\{1,\dots,k-1\}} \max \biggl\{ \frac{\max\{x_1,k\}}{2}, \min\{x_1,k\} \biggr\} + k \\
	={}& \beta_{k-1} + k + 2\sum\nolimits_{x_1\in\{1,dots,k-1\}} \max \biggl\{ \frac{k}{2},x_1 \biggr\}\\
	={}& \beta_{k-1} + \frac{5}{4}k^2 + o(k^2)
\end{align*}
Since $\beta_1=1$,
\begin{equation*}
	\beta_k = 1 + \sum_{\ell=1}^k \biggl( \frac{5}{4}k^2 + o(k^2) \biggr) = \frac{5}{12}k^3 + o(k^3)
\end{equation*}
and thus
\begin{equation*}
	\alpha_k(2) = \frac{7}{12} + \frac{o(k^3)}{k^3} .
\end{equation*}
This expression tends to $7/12$ as $k$ tends to infinity, and the claim follows by the same argument as in the proof of \thmref{thm:permutation}.
\end{proof}

One may wonder whether this bounds is tight, for the permutation mechanism or even in general. We leave this as an open question, but conclude by giving upper bounds of $2/3$ and $3/4$, respectively, on possible values of $\alpha$ for the permutation mechanism and any impartial mechanism.

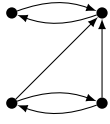
\begin{figure}[tb]
	\centering
	\begin{tikzpicture}[bend angle=20]
		\tikzstyle{every circle node}=[fill,inner sep=1.5pt,outer sep=0pt]
		\draw (0,0) node(1)[circle]{} ++(1.2,0) node(2)[circle]{} ++(0,-1.2) node(3)[circle]{} ++(-1.2,0) node(4)[circle]{};
		\foreach \x/\y in {3/2,4/2} {\draw[-latex] (\x) to (\y);}
		\foreach \x/\y in {1/2,2/1,3/4,4/3} {\draw[-latex] (\x) to[bend left] (\y);}
	\end{tikzpicture}
	\caption{A graph on which the permutation mechanism is $2/3$-optimal}
	\label{fig:perm_up}
\end{figure}
To see that the permutation mechanism cannot be more than $2/3$-optimal, consider the graph of \figref{fig:perm_up}. The unique vertex with degree~$3$ in this graph is selected by the permutation mechanism if and only if it appears in the last two positions of the permutation, which happens with probability $1/2$. Indeed, when it appears in one of the first two positions it has degree at most~$1$ at the time it is considered by the mechanism. At the same time, one of the vertices in the last two positions has degree~$1$ when it is considered and consequently gets selected. The expected degree of the selected vertex is thus $3\cdot 1/2+ 1\cdot 1/2=2$, compared to a maximum degree of~$3$. What is interesting about this bound is that it is attained for a graph with maximum degree~$3$. This suggests that a matching lower bound could not be obtained from a monotonicity result like that of \thmref{thm:alpha_k_mono}.

	The same upper bound of $2/3$ holds asymptotically for the more restricted case considered by \citet{HoMo13a} where every vertex has outdegree one. To see this, consider the graph $G=(N,E)$ with $N=\{1,\dots,n\}$ and $E=\{(i,i+1)\midd i=1,\dots,n-2\}\cup\{(n-1,1),(n,1)\}$, and observe that the permutation mechanism selects vertex~$1$, the unique vertex with degree~$2$, with significant probability only for permutations in which vertices~$n-1$ and~$n$ both occur before~$1$. Since the latter happens with probability exactly $1/3$, the expected degree of the selected vertex is not significantly greater than $2\cdot 1/3+ 1\cdot 2/3=4/3$, compared to a maximum degree of~$2$.

	Our final result establishes upper bounds on $\alpha$ for \emph{any} mechanism that is impartial and $\alpha$-optimal on $\G^{+}$, and for different values of~$n$. Intuitively these bounds arise as dual solutions of an optimization problem characterizing the $\alpha$-optimal impartial mechanisms for the maximum value of $\alpha$. These dual solutions are optimal, and the upper bound therefore tight, for $n\leq 7$.
\begin{theorem}  \label{thm:upper}
	Consider an impartial selection mechanism that is $\alpha$-optimal on $\G^{+}_n$. Then
	\begin{equation*}
		\alpha \geq \begin{cases}
		 	3 / 4 & \text{if $n=3$,} \\
			(3n-1) / 4n & \text{otherwise.}
		\end{cases}
	\end{equation*}
\end{theorem}	
\begin{proof}
	By \lemref{lem:symmetric} we can restrict our attention to symmetric mechanisms.
	
	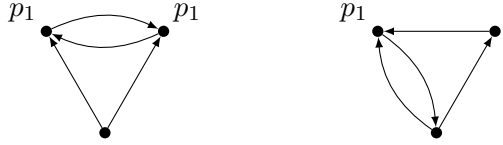
\begin{figure}[tb]
		\centering
		\begin{tikzpicture}[bend angle=25]
			\tikzstyle{every circle node}=[fill,inner sep=1.5pt,outer sep=0pt]
			\newcommand{\tri}[5]{\draw (0,0) +(150:.9) node(1)[circle]{} node[above left]{#3} +(30:.9)
node(2)[circle]{} node[above right]{#4} +(270:.9) node(3)[circle]{} node[below]{#5};
			\foreach \x/\y in {#1} {\draw[-latex] (\x) to (\y);}
			\foreach \x/\y in {#2} {\draw[-latex] (\x) to[bend left] (\y);}}
			\matrix[column sep=1.6cm,row sep=1.4cm]{
			\tri{3/1,3/2}{1/2,2/1}{$p_1$}{$p_1$}{} &
			\tri{2/1,3/2}{1/3,3/1}{$p_1$}{}{} \\};
		\end{tikzpicture}
		\caption{Impartial probability assignment for two graphs with $n=3$}
		\label{fig:oneplus3}
	\end{figure}
	First assume that $n=3$, and consider the two graphs shown in \figref{fig:oneplus3}. It is easily verified that any impartial mechanism must assign probabilities as shown, and it must therefore be the case that $p_1\leq\frac{1}{2}$. In the
graph on the right, the agent assigned probability $p_1$ is the unique agent with the maximum indegree of~$2$, and thus 
	\[
		\alpha \leq \frac{2 p_1 + (1-p_1)}{2} = \frac{p_1 + 1}{2} \leq \frac{3}{4} .
	\]
	
	\begin{figure}[tb]
		\centering
		\begin{tikzpicture}[bend angle=25]
			\tikzstyle{every circle node}=[fill,inner sep=1.5pt,outer sep=0pt]
			\newcommand{\sqrx}[6]{\draw (0,0) node(1)[circle]{} node[above left]{#3} ++(1.2,0) node(2)[circle]{}
node[above right]{#4} 
			++(0,-1.2) node(3)[circle]{} node[below right]{#5} 
			++(-1.2,0) node(4)[circle]{} node[below left]{#6} ++(0,0.52);
			\foreach \x/\y in {#1} {\draw[-latex] (\x) to (\y);}
			\foreach \x/\y in {#2} {\draw[-latex] (\x) to[bend left] (\y);}}
			\matrix[column sep=1.4cm,row sep=1.4cm]{
			\sqrx{}{1/2,2/1,3/4,4/3}{$p_1$}{$p_1$}{$p_1$}{$p_1$} &
			\sqrx{3/4,4/1,4/2}{1/2,2/1}{$p_2$}{$p_2$}{}{$p_1$} &
			\sqrx{1/2,2/3,3/4,4/1,4/2}{}{}{$p_2$}{}{} \\};
		\end{tikzpicture}
		\caption{Impartial probability assignment for three graphs with $n=4$}
		\label{fig:oneplus4}
	\end{figure}
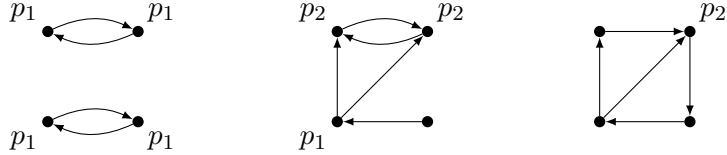
	Now asume that $n\geq 4$ even, and consider the set of three graphs on~$n$ agents where agents~$1$ to~$4$ vote as in
\figref{fig:oneplus4} and the remaining $n-4$ agents are grouped in pairs such that agent $2i$ votes for agent $2i-1$ and
vice versa. It is easily verified that any impartial mechanism must assign probabilities as in \figref{fig:oneplus4}, and thus
$np_1=1$ and $p_1+2p_2\leq 1$. Moreover, the agent assigned probability $p_2$ in the rightmost graph is the unique agent
with indegree $2$ in that graph, and thus
	\[
		\alpha \leq \frac{2p_2+(1-p_2)}{2} = \frac{p_2+1}{2} \leq \frac{\frac{n-1}{2n}+1}{2} = \frac{3n-1}{4n} .
	\]
	
	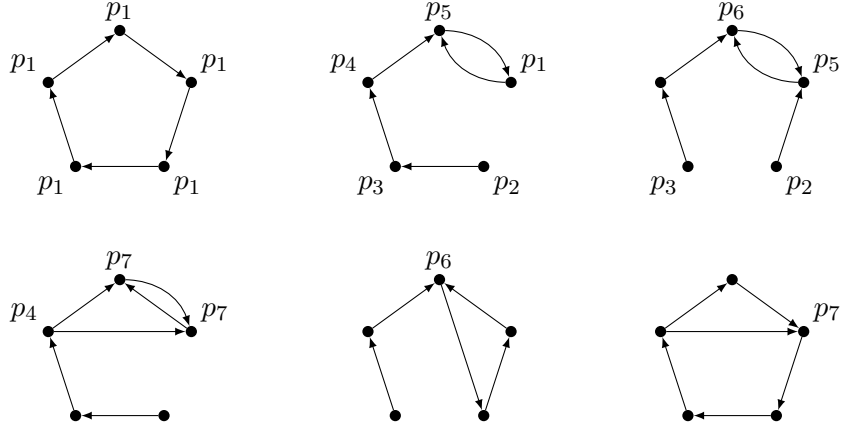
\begin{figure}[tb]
		\centering
		\begin{tikzpicture}[bend angle=35]
			\tikzstyle{every circle node}=[fill,inner sep=1.5pt,outer sep=0pt]
			\newcommand{\pent}[7]{\draw (0,0) +(162:1) node(1)[circle]{} node[above left]{#3} +(90:1)
node(2)[circle]{} node[above]{#4} +(18:1) node(3)[circle]{} node[above right]{#5} +(306:1) node(4)[circle]{} node[below
right]{#6} +(234:1) node(5)[circle]{} node[below left]{#7};
			\foreach \x/\y in {#1} {\draw[-latex] (\x) to (\y);}
			\foreach \x/\y in {#2} {\draw[-latex] (\x) to[bend left] (\y);}}
			\matrix[column sep=1.1cm,row sep=.5cm]{
			\pent{1/2,2/3,3/4,4/5,5/1}{}{$p_1$}{$p_1$}{$p_1$}{$p_1$}{$p_1$} &
			\pent{4/5,5/1,1/2}{2/3,3/2}{$p_4$}{$p_5$}{$p_1$}{$p_2$}{$p_3$} &
			\pent{5/1,1/2,4/3}{2/3,3/2}{}{$p_6$}{$p_5$}{$p_2$}{$p_3$} \\
			\pent{4/5,5/1,1/2,1/3,3/2}{2/3}{$p_4$}{$p_7$}{$p_7$}{}{} &
			\pent{5/1,1/2,2/4,4/3,3/2}{}{}{$p_6$}{}{}{} &
			\pent{1/2,1/3,2/3,3/4,4/5,5/1}{}{}{}{$p_7$}{}{} \\};
		\end{tikzpicture}
		\caption{Impartial probability assignment for six graphs with $n=5$}
		\label{fig:oneplus5}
	\end{figure}
	Now asume that $n=5$, and consider the six graphs shown in \figref{fig:oneplus5}. It is easily verified that any impartial mechanism must assign probabilities as in \figref{fig:oneplus5}, so
	\begin{align}
		p_1 &= 1/5, \label{eq:d1} \\
		p_1 + p_2 + p_3 + p_4 + p_5 &= 1, \label{eq:d2} \\
		p_2 + p_3 + p_5 + p_6 &\leq 1, \label{eq:d3} \\
		p_4 + 2p_7 &\leq 1 . \label{eq:d4}
	\end{align}
	By adding \eqref{eq:d1}, \eqref{eq:d3}, and \eqref{eq:d4} and subtracting~\eqref{eq:d2},
	\[
		p_6 + 2p_7 \leq \frac{6}{5} \qquad\text{and thus}\qquad \min(p_6,p_7) \leq \frac{2}{5}.
	\]
	The agents assigned probabilities $p_6$ and $p_7$ in the two rightmost graphs in the bottom row of
\figref{fig:oneplus5} are the unique agents with indegree $2$ in those graphs, so
	\[
		\alpha \leq \frac{2p_6+(1-p_6)}{2} = \frac{p_6+1}{2} \qquad\text{and}\qquad \alpha \leq
\frac{2p_6+(1-p_6)}{2} = \frac{p_6+1}{2},
	\]
	and thus
	\[
		\alpha \leq \frac{\min(p_6,p_7)+1}{2} \leq \frac{\frac{2}{5}+1}{2} = \frac{7}{10} = \frac{3n-1}{4n}.
	\]
	
	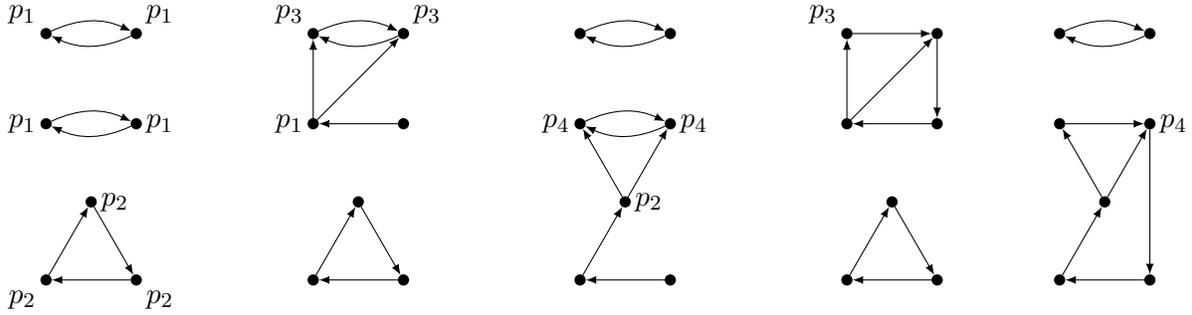
\begin{figure}[tb]
		\centering
		\begin{tikzpicture}[bend angle=25]
			\tikzstyle{every circle node}=[fill,inner sep=1.5pt,outer sep=0pt]
			\newcommand{\hept}[9]{\draw (0,0) node(1)[circle]{} node[above left]{#3}
			++(0:1.2) node(2)[circle]{} node[above right]{#4} ++(-90:1.2) node(4)[circle]{} node[right]{#6}
++(180:1.2) node(3)[circle]{} node[left]{#5} ++(-60:1.2) node(5)[circle]{} node[right]{#7} ++(-60:1.2) node(7)[circle]{}
node[below right]{#9} ++(-180:1.2) node(6)[circle]{} node[below left]{#8};
			\foreach \x/\y in {#1} {\draw[-latex] (\x) to (\y);}
			\foreach \x/\y in {#2} {\draw[-latex] (\x) to[bend left] (\y);}}
			\matrix[column sep=1.1cm,row sep=1.4cm]{
			\hept{5/7,7/6,6/5}{1/2,2/1,3/4,4/3}{$p_1$}{$p_1$}{$p_1$}{$p_1$}{$p_2$}{$p_2$}{$p_2$} & 
			\hept{3/1,3/2,4/3,5/7,7/6,6/5}{1/2,2/1}{$p_3$}{$p_3$}{$p_1$}{}{}{}{} & 
			\hept{5/3,5/4,7/6,6/5}{1/2,2/1,3/4,4/3}{}{}{$p_4$}{$p_4$}{$p_2$}{}{} &
			\hept{1/2,2/4,4/3,3/1,3/2,5/7,7/6,6/5}{}{$p_3$}{}{}{}{}{}{} & 
			\hept{3/4,4/7,5/3,5/4,6/5,7/6}{1/2,2/1}{}{}{}{$p_4$}{}{}{} \\ };
		\end{tikzpicture}
		\caption{Impartial probability assignment for five graphs with $n=7$}
		\label{fig:oneplus7}
	\end{figure}
	Finally assume that $n\geq 7$ odd, and consider the set of five graphs on~$n$ agents where agents~$1$ to~$7$ vote as
in \figref{fig:oneplus7} and the remaining $n-7$ agents are grouped in pairs such that agent $2i$ votes for agent $2i-1$ and
vice versa. It is easily verified that any impartial mechanism must assign probabilities as in \figref{fig:oneplus7}, so
	\begin{align*}
		(n-3)p_1 + 3p_2 &=1, \\
		p_1 + 2p_3 &\leq 1, \\
		p_2 + 2p_4 &\leq 1 .
	\end{align*}
	The agents assigned probabilities $p_3$ and $p_4$ in the two rightmost graphs are the unique agents with indegree
$2$ in those graphs, so
	\begin{align*}
		\alpha &\leq \frac{2p_3+(1-p_3)}{2} = \frac{p_3+1}{2} \leq
		\frac{\frac{1-p_1}{2}+1}{2} = \frac{3-p_1}{4} , \\
		\alpha &\leq \frac{2p_4+(1-p_4)}{2} = \frac{p_4+1}{2} \leq
		\frac{\frac{1-p_2}{2}+1}{2} = \frac{3-p_2}{4} ,
	\end{align*}
	and thus
	\[
		\alpha \leq \frac{3-\max(p_1,p_2)}{4} \leq \frac{3-\frac{1}{n}}{4} = \frac{3n-1}{4n} ,
	\]
	where the second inequality holds because $\max(p_1,p_2)\geq 1/n$.
\end{proof}

Somewhat surprisingly, restricting the set of graphs even further, by requiring that every vertex has outdegree \emph{exactly}~$1$, does not enable significantly better impartial mechanism. Using similar arguments as in the proof of \thmref{thm:upper}, it can be shown that in this case any impartial and $\alpha$-optimal mechanism must satisfy $\alpha\leq 5/6$ if $n=3$, $\alpha\leq (6n-1)/8n$ if $n\geq 6$ and even, and $\alpha\leq 3/4$ otherwise. These bounds are tight for $n\leq 9$.

\end{document}